\documentclass[letterpaper, 10 pt, conference]{ieeeconf}
\IEEEoverridecommandlockouts
\usepackage{cite}
\usepackage{amsmath,amssymb,amsfonts,amsthm, nccmath}
\usepackage{graphicx}
\usepackage{cases}
\usepackage{xcolor, soul}
\usepackage{algorithm, algpseudocode}
\usepackage{url}

\usepackage[colorinlistoftodos]{todonotes}
\usepackage{textcomp}

\theoremstyle{plain}

\newtheorem{prop}{Proposition}

\theoremstyle{definition}

\newtheorem{ex}{Example}
\newtheorem*{rem*}{Remark}

\renewcommand{\rm}[1]{\mathrm{#1}}
\renewcommand{\bf}[1]{\mathbf{#1}}
\newcommand{\bb}[1]{\mathbb{#1}}

\newcommand{\R}{\bb{R}}

\renewcommand{\o}{\rm{o}_t}
\renewcommand{\c}{\rm{c}_t}
\newcommand{\h}{\rm{h}_t}

\begin{document}
\title{Efficient Industrial Refrigeration Scheduling with Peak Pricing}

\author{Rohit Konda, Jordan Prescott, Vikas Chandan, Jesse Crossno, \\ Blake Pollard, Dan Walsh, Rick Bohonek, and Jason R. Marden \thanks{R. Konda (\texttt{rkonda@ucsb.edu}) and J. R. Marden (\texttt{jrmarden@ece.ucsb.edu})are with the Department of Electrical and Computer Engineering at the University of California, Santa Barbara, CA. Vikas Chandan \texttt{vikas@crossnokaye.com}, Jesse Crossno \texttt{crossno@crossnokaye.com}, Blake Pollard \texttt{blake@crossnokaye.com}, and Dan Walsh \texttt{dan@crossnokaye.com} are with CrossnoKaye\textregistered . Rick Bohonek \texttt{rbohonek@butterball.com} is with Butterball LLC\textregistered. This work is supported by funding from CrossnoKaye.}}

\maketitle
\thispagestyle{empty}

\begin{abstract}
The widespread use of industrial refrigeration systems across various sectors contribute significantly to global energy consumption, highlighting substantial opportunities for energy conservation through intelligent control design. As such, this work focuses on control algorithm design in industrial refrigeration that minimize operational costs and provide efficient heat extraction. By adopting tools from inventory control, we characterize the structure of these optimal control policies, exploring the impact of different energy cost-rate structures such as time-of-use (TOU) pricing and peak pricing. While classical threshold policies are optimal under TOU costs, introducing peak pricing challenges their optimality, emphasizing the need for carefully designed control strategies in the presence of significant peak costs. We provide theoretical findings and simulation studies on this phenomenon, offering insights for more efficient industrial refrigeration management.
\end{abstract}

\section{Introduction}
\label{sec:int}

Industrial refrigeration systems are widely utilized across diverse sectors, not limited to plastics manufacturing, chemical processing, food storage, and electronics production \cite{dincer2017refrigeration, stoecker1998industrial, fabrega2010exergetic, tassou2010review}. Collectively, industrial refrigeration contributes to approximately $8.4\%$ of total energy consumption in the United States \cite{eiareport}. Consequently, there exist significant opportunities for energy conservation within industrial refrigeration. These opportunities extend beyond hardware upgrades to include enhancements in the control algorithms implemented in these systems. Improvement in the algorithm design can potentially be more appealing, as they can yield substantial energy savings with minimal capital investment required for system retrofitting.

The intention behind improvements in algorithm design is to strategically adjust the control strategies of the components within the industrial refrigeration process to achieve the necessary heat extraction while minimizing operation costs. These costs can be measured through total power, electric costs, carbon emissions, or other relevant metrics. For example, the set points or steady state configurations of the components can be optimized to raise the energy efficiency \cite{zhao2013model, larsen2003control, larsen2006model, manske2000performance}. Furthermore, model predictive control or trajectory optimization can be utilized to dynamically optimize the energy efficiency of the refrigeration cycle over a time horizon \cite{hovgaard2013nonconvex, yin2018model, shafiei2014model}. Another approach that has garnered significant attention is \emph{thermal load shifting}, where refrigeration loads are dynamically managed to leverage variable energy cost-rate structures \cite{sun2013peak, yao2021state, pattison2016optimal, pattison2017moving, vishwanath2019iot}. While these results are quite encouraging in conventional cost structures, there remains a need to characterize the qualitative behavior of optimal control policies when considering more varied cost structures.

There exist many different energy cost-rate structures that may depend on various factors. \emph{Fixed rate} pricing, the most traditional rate structure, involves charging a flat rate per unit of energy consumed, irrespective of external conditions. In contrast, \emph{time-of-use} (TOU) pricing varies the per-unit rate based on the time or season. Typically, this results in higher energy costs during peak demand periods and lower energy costs during off-peak periods. In this setup, there may be significant economic benefit for the operator to shift their energy usage to off-peak hours. Deciding on how to shift is exactly the focus of the thermal load shifting literature. To further regulate energy usage, \emph{peak pricing} can also be introduced, where the cost is dependent on the maximum power usage of the system over a time period and spikes in energy consumption are highly disincentived. In many industrial refrigeration systems, peak pricing may comprise of a large portion of the energy costs. In this work, we focus on when both peak pricing and TOU costs are present and quantify the impact on optimizing the scheduling of refrigeration loads. 

\begin{figure}[h!]
    \centering
    \includegraphics[width=230pt]{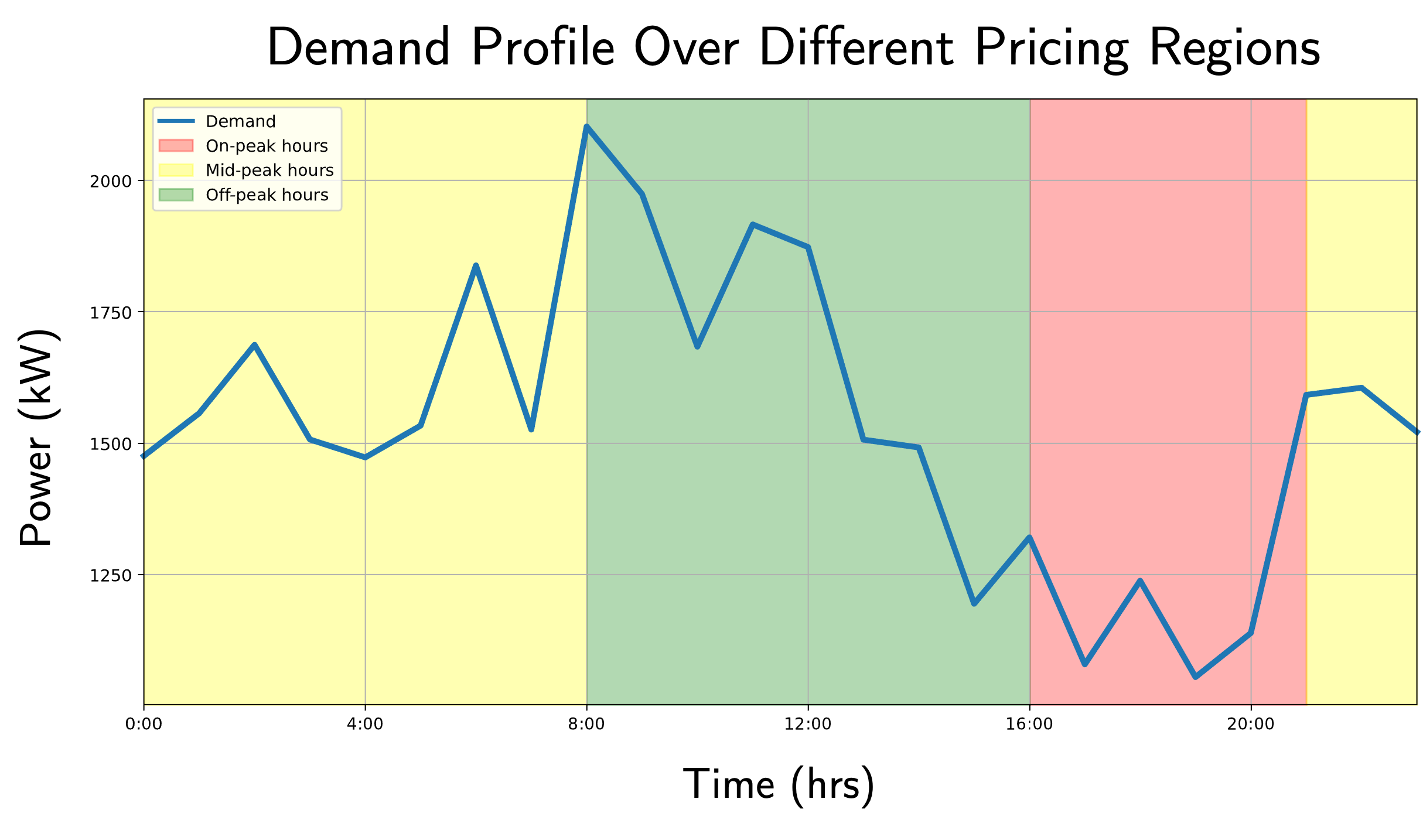}
    \caption{We depict a prototypical power consumption profile over different pricing regions.}
    \label{fig:peak}
\end{figure}

\noindent \textbf{Case Study.} We introduce an example of a rate structure that a refrigeration facility may be charged with for its energy consumption.\footnote{This data was obtained through a direct collaboration with CrossnoKaye (see \url{crossnokaye.com}), a company that focuses in integrating control systems for industrial refrigeration systems within the cold food and beverage sector.} The time-of-use costs are dependent on time of day, where the cost rates vary between three time windows: on-peak hours (4:00PM - 9:00PM), mid-peak hours (9:00PM - 8:00AM), and off-peak hours (8:00AM - 4:00PM). Typically, energy-costs are higher in on-peak hours and lower in off-peak hours. Additionally, the facility can incur additional peak costs: a peak charge over the maximum energy consumption over a month and a peak charge over the month for each of the specific time windows (on, mid, off). We display a possible demand profile under this cost structure in Figure \ref{fig:peak}.
\vspace{4pt}

Designing control policies using optimization methods that account for peak pricing, though not as common as TOU pricing, has been studied previously \cite{risbeck2019economic, mo2021optimal, oldewurtel2010reducing}. However, the primary objective of this work is to theoretically evaluate how peak pricing influences the qualitative structure of the optimal control policies. We approach this through uniquely adopting the perspective of \emph{inventory control}. Inventory control \cite{axsater2015inventory, liu2012decision} is a classical branch of multi-stage decision problems that explores purchasing policies of inventory to ensure optimal warehouse stock levels. When energy rates solely consist of time-of-use (TOU) costs, classical findings from inventory control suggest that the optimal control algorithms should adopt a threshold approach: when the facility's temperature exceeds a specified threshold, the refrigeration load is increased to return the facility to a desired buffer temperature. Threshold policies are commonly implemented in practice for industrial refrigeration; we verify its optimality under TOU costs in Proposition \ref{prop:nopeak}. However, the introduction of peak pricing disrupts the optimality of threshold policies. In fact, in Proposition \ref{prop:main}, we characterize the structure of the optimal control policy under peak pricing; we see that the optimality of simple threshold policies is lost in these settings. We also verify our findings through simulation in Section \ref{sec:simulations}. Our work suggests that designing control policies should be done carefully if significant peak costs are present.

\section{Peak Pricing Model}
\label{sec:model}

In this paper, we aim to devise scheduling strategies for refrigeration systems that strike a balance between minimizing overall costs and adequately meeting the necessary cooling demands of the facility. Additionally, we investigate the impact of \emph{peak pricing}, a common energy pricing mechanism involving charging based on the maximum energy consumption over a specific period. While such pricing mechanisms are commonly employed, strategies that directly account for these costs are not as well explored. Therefore we leverage concepts from \emph{inventory control} to analyze how peak pricing influences the optimal scheduling policies. 

To focus on this, we simplify the operational process of the industrial refrigeration, and solely focus on the relationship between the total cooling and energy expenditures of the refrigeration process. The primary system state is the facility temperature, denoted as $x_t \in \R$ for a given time $t \in \{1, \dots, T\}$. Here, $T$ represents the length of the horizon under consideration. Using a first-order model of specific heat\footnote{According to the specific heat equation, $x_{t+1} - x_t = C_f Q_{\rm{net}}$, where $Q_{\rm{net}}$ is the net heat transfer and $C_f$ is the heat capacity. For simplicity of notation, we assume that $C_f = 1$ for the manuscript.}, we can succinctly describe the dynamics of the facility's temperature as follows:
\begin{equation}
    x_{t+1} = x_t - u_t + q_t,
\end{equation}
where $u_t \geq 0$ represents the heat removed from the facility via the refrigeration system and $q_t$ represents the heat influx from the surrounding environment. We assume that over the horizon, the incoming heat $q_t \geq 0$ is a non-negative random variable that is drawn from a known distribution $\bf{Q}_t$ that may be time varying. As an illustration, we present a scatter plot of possible heat demands over a day\footnote{Over June 2023, power and coefficient of power (COP) estimates were collected from a refrigeration site of Butterball LLC \textregistered \ to estimate the incoming refrigeration loads. } in Figure \ref{fig:Qt}.

\begin{figure}[h!]
    \centering
    \includegraphics[width=240pt]{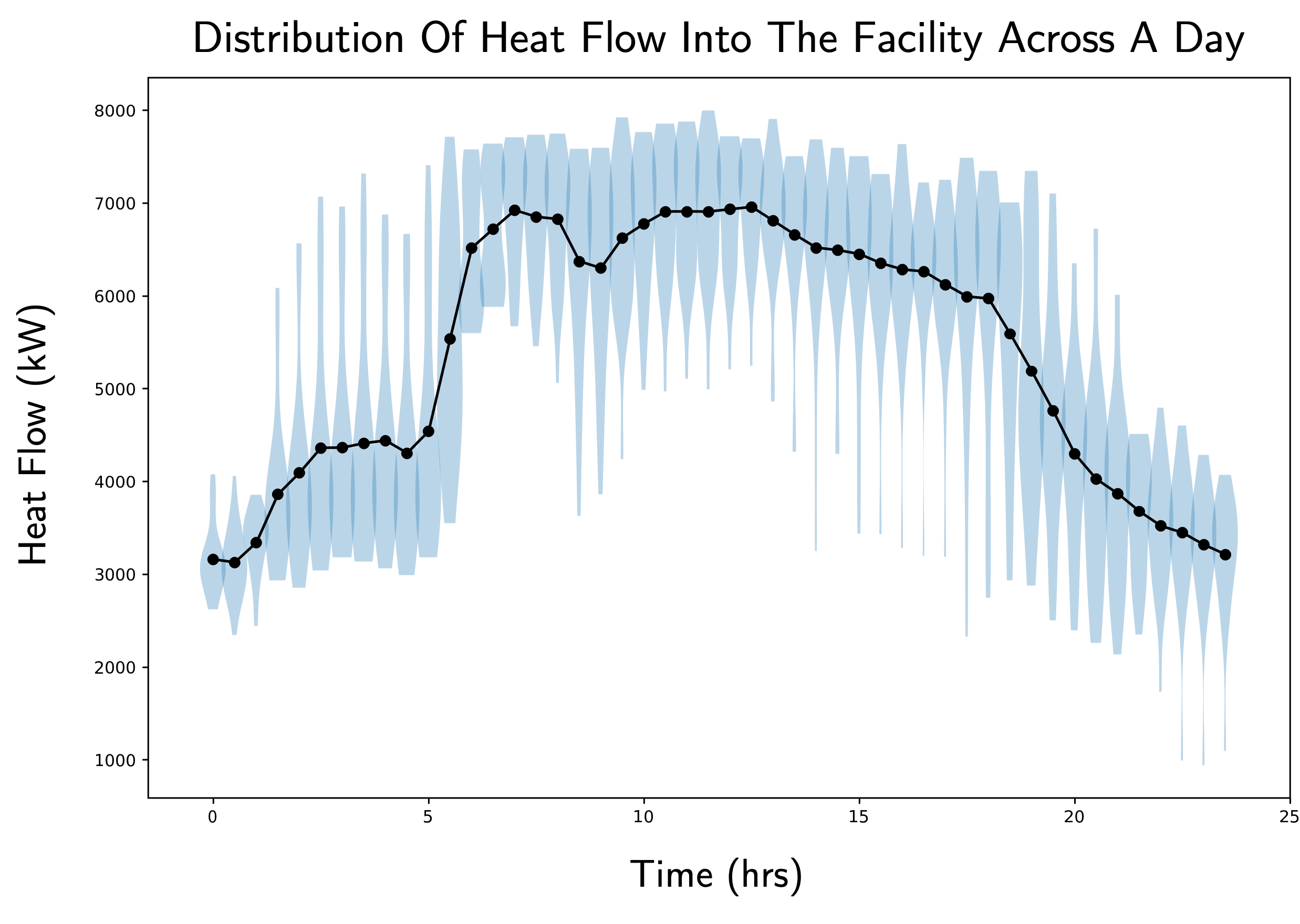}
    \caption{We present distributions and respective averages of $\bf{Q}_t$ over the horizon of a day for a particular refrigeration facility. }
    \label{fig:Qt}
\end{figure}

\begin{figure}[t!]
    \centering
    \includegraphics[width=200pt]{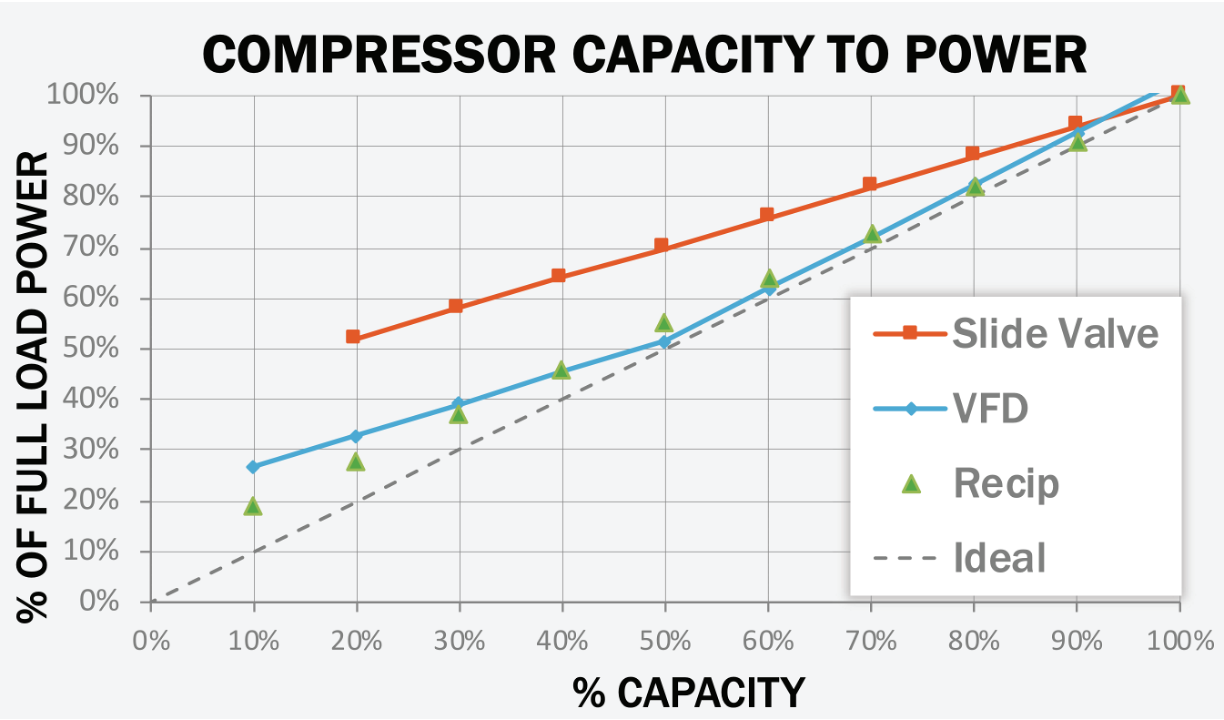}
    \caption{We display the power-heat curves of different compressor types as taken from \cite{factandfig}. We see that for compressors without variable frequency drives (VFDs), the power draw and the respective thermal capacity share an affine relationship.}
    \label{fig:ordercost}
\end{figure}

The control task of the refrigeration system is to generate a sequence of refrigeration loads $\bf{u} \equiv u_1, \dots, u_T$ that minimizes the overall system cost, dependent on power consumption, while providing the necessary heat extraction. For a food storage facility, the required refrigeration comes in the form of a temperature constraint $x_t \leq 0$, where we would like to maintain temperatures to be below freezing. We capture both the energy costs and temperature constraint violations in the following stage cost,
\begin{equation}
\label{eq:stage}
    \c(x_t, u_t) = \o(u_t) + \bb{E}_{q_t} \left[ \h(x_{t+1}) \right].
\end{equation}
Here, $\o: \R_{\geq 0} \to \R_{\geq 0}$ describes the cost associated with running the refrigeration at a load of $u_t$ and $\h: \R \to \R_{\geq 0}$ represents the penalty costs for set-point deviations from the desired temperature $x_t = 0$. For each $t$, we assume that $\h$ is continuous, convex, and has a minimum at $0$, i.e. our desired temperature set point. If $x > 0$ is positive, $\h(x)$ represents the penalty cost for temperature constraint violations. If $x < 0$ is negative, then $\h(x)$ represents the cost of excessive cooling which may lead to system inefficiencies. An example of a reasonable penalty function is
\begin{equation}
\label{eq:viol}
    \h(x) = \begin{cases}
        b_t \cdot x^2 & \text{ if } x \geq 0 \\
        d_t \cdot x^2 & \text{ if } x < 0.
    \end{cases}
\end{equation}
with $b_t \gg d_t$ for each $t$. Since the energy costs are proportional to the total power consumption of the refrigeration system, we assume an affine structure for the energy costs with respect to a given refrigeration load, written as
\begin{equation}
\label{eq:ordering}
    \o(u_t) = \begin{cases}
        K + a_t \cdot u_t & \text{ if } u_t > 0 \\
        0 & \text{ if } u_t = 0,
    \end{cases}
\end{equation}
where $K$ represents the setup cost of having the refrigeration system in operation and $a_t$ determines the per-unit cost for refrigeration capacity at time $t$. The per-unit cost $a_t$ may depend heavily on $t$, representing a potential time-of-use cost structure. This cost model is reflected in the power-heat curves of the compressors, which represent the majority ($\sim 90 \%$) of the energy consumption of the refrigeration process, as shown in Figure \ref{fig:ordercost}. 

The cost for peak pricing is reflected in the maximum refrigeration load over the horizon. More formally, the peak price can be written as $P \cdot \max \{u_t \text{ for } t \leq T\}$, where $P \geq 0$ is the scaling factor associated with the peak costs. To represent peak cost as a terminal cost, we can introduce an auxiliary state variable $y \in \R_{\geq 0}$ with dynamics 
\begin{equation}
    y_{t+1} = \max \{y_t, u_t \}.
\end{equation}
We note that while the structure of problem without peak pricing is classical in the inventory control literature (for this, see \cite{porteus1990stochastic, lu2014inventory, sobel1970making} for similar results with regards to convex, piecewise affine $\o$), the addition of peak pricing is a novel consideration. While not as studied, peak pricing is extremely significant to the cost structures for industrial processes.

Consolidating the stage and terminal costs, the total cost of refrigeration process under the sequence of loads $\bf{u}$ with a given peak $y$ is 
\begin{equation}
\label{eq:totalcost}
    J(x, y, \bf{u}) = \bb{E} \left[\sum_{t=1}^{T} \c(x_t, u_t) \right] + P \max \Big\{ y, \{u\}_{t \leq T} \Big\},
\end{equation}
where $x_t$ follows from the respective state transition probabilities from the initial state $x_1 = x$ and the expectation is taken over the possible incoming heat $q_t \sim \bf{Q}_t$. The optimal total cost can be written recursively via the Bellman equation,
\begin{align}
    &V_{t}(x, y) = \min_{u \geq 0} \Big\{ c_t(x, u) + \bb{E}_{q_t} \Big[ V_{t+1}(x^+ + q_t, y^+)\Big] \Big\}, \label{eq:VT}\\
\label{eq:Vterm}
    &V_{T+1}(x, y) = P \cdot y.
\end{align}
$V_{T+1}$ represents the terminal cost of the dynamic program. We use $x^+ = x - u$ and $ y^+ = \max\{y, u\}$ to denote the successor states for simplicity of notation. The optimal loads can be written in feedback form with a policy function $\pi^*_t: \R \times \R_{\geq 0} \to \R_{\geq 0}$ as the argument to the previous optimization formulation:
\begin{equation}
    \pi^*_t(x, y) \in \arg \min_{u \geq 0} \left\{ c_t(x, u) + \bb{E}_{q_t}\left[ V_{t+1}(x^+, y^+)\right] \right\}.
\end{equation}

The main concern of this work is on characterizing the structure of these optimal policies with respect to peak pricing. We do this through analytical characterizations in Section \ref{sec:policy} and through simulations in Section \ref{sec:simulations}.

\section{Optimal Policy Characterizations}
\label{sec:policy}

We first characterize the structure of optimal refrigeration policies with no peak pricing costs, i.e. when $P = 0$ in Eq. \eqref{eq:totalcost}, where there are only TOU costs present. In this case, our models align with the standard ones present in inventory control, and we can directly invoke classical results to get the structure of the optimal policies. Interestingly, the optimal policies simplify to a threshold strategy.

\begin{prop}[\cite{scarf1960optimality}]
\label{prop:nopeak} 
Consider the industrial refrigeration problem with a total cost in Eq. \eqref{eq:totalcost} with no peak cost ($P = 0$). The optimal policy $\pi^*_t$ is a threshold policy of the form
\begin{equation}
\label{eq:threshold}
  \pi^*_t(x, y) =
\begin{cases}
x - S_t & \text{ if } x > s_t, \\
0 & \text{ if } x \leq s_t,
\end{cases}
\end{equation}
for some $s_t \geq S_t \in \R$ for every $t \leq T$.
\end{prop}
\begin{proof}
We can directly use classical results in inventory control with sign changes (for e.g., see \cite{scarf1960optimality}, \cite[Chapter 2.6]{liu2012decision}, or \cite[Chapter~4]{song2023research}) to show the claim.
\end{proof}

From the above proposition, we see that the optimal structure comes in the form of a simple threshold policy, where the refrigeration system is turned off if the facility is below a certain buffer temperature $s_t$ but if turned on, sets the facility to a lower buffer temperature $S_t$. We remark that if there is no setup costs, or that $K = 0$ in Eq. \eqref{eq:totalcost}, the two buffer temperatures align, with $s_t = S_t$.

Threshold policies are commonly implemented in industrial refrigeration to schedule refrigeration loads. While we have shown that these threshold policies are optimal without peak pricing, computing closed form solutions for $s_t$ and $S_t$ is not possible in general. However, these buffer temperatures can be derived through data-driven strategies to produce well-performing threshold parameters. We do this in Section \ref{sec:simulations}.

When peak pricing is introduced to the total cost, however, we will see that the optimal policies stop corresponding to threshold policies. In fact, for simple models of refrigeration, we see that the optimal policies can become quite complicated in the next example.

\begin{ex}[Peak Pricing]
\label{ex:peak}
Consider a toy refrigeration scenario, where there is external heat input $q_t \equiv 0$ for all $t$. Let the temperature penalty function be $\h(x) = b | x |$ for all $t \leq T$. We set the parameters $P \geq b \geq a = 1$ for the peak, penalty and refrigeration costs, and fix the setup cost to $K=0$ to frame the problem naturally. Even within this basic model, the optimal policy can be surprisingly intricate.

To see this, we solve for the value function explicitly in Eq. \eqref{eq:VT} via the backwards recursion. Under the model assumptions, the value function $V_T$ at time $T$ can be simply expressed as
\begin{equation*}
    V_T(x, y) = \min_{u \geq 0} \Big\{u + b| x - u | + P \max \{y, u\} \Big\}.
\end{equation*}
From the value function, there are three regimes to describe the optimal policy. These three regimes are depicted in Figure \ref{fig:abs}, where the optimal policy is explicitly characterized below.
\begin{equation*}
    \pi^*_{T}(x, y) = \begin{cases}
        0 & \text{ if } x \leq 0, \\
        x & \text{ if } 0 \leq x \leq y, \\
        y & \text{ if }  x \geq y.
    \end{cases}
\end{equation*}

\begin{figure}[h]
    \centering
    \includegraphics[width=200pt]{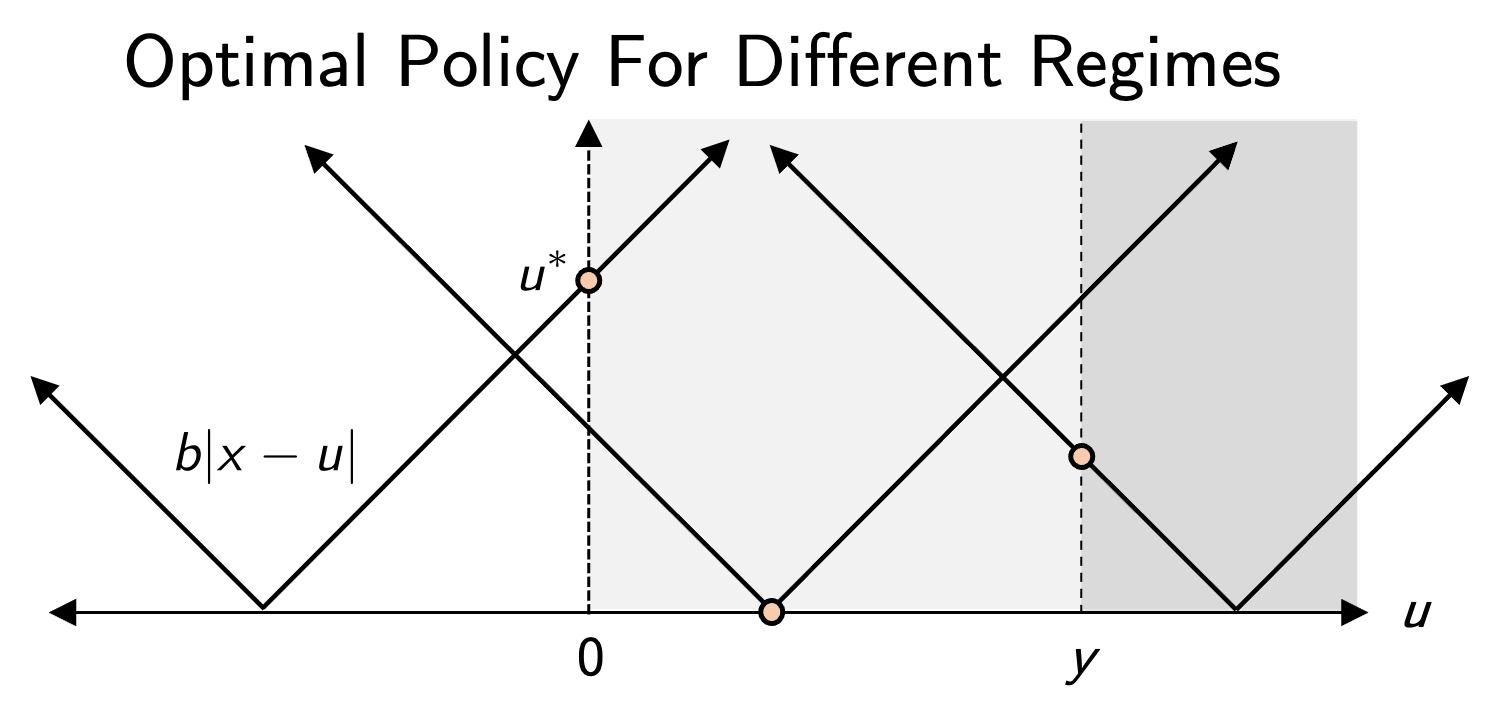}
    \caption{We depict the possible optimal inputs for each of the three regimes, dependent on $x$. Note that since $P \geq b$, the optimal $u^*$ must live between $[0, y]$, as delineated by the orange markers.}
    \label{fig:abs}
\end{figure}

The optimal policy across one time step is a threshold policy with a cap at $u = y$ due to the peak cost. However, adding another step in the backwards recursion complicates the optimal policy $\pi^*_{T-1}$. For the next step in the value recursion, we have that
\begin{equation*}
    V_{T-1}(x, y) = \min_{u \geq 0} \Big\{u + b| x - u | + V_{T}(x^+, y^+) \Big\},
\end{equation*}
where the value function $V_T$ can be described in closed form via the characterization of the optimal policy $\pi^*_T$ as
\begin{equation*}
    V_T(x, y) = \begin{cases}
        P y - b x & \text{ if } x \leq 0, \\
        P y + x & \text{ if } 0 \leq x \leq y, \\
        (P + 1 - b)y + b x & \text{ if }  x \geq y.
    \end{cases}
\end{equation*}

With the expression of the value function $V_T$, we can describe the optimal policy for a horizon of $2$. Since $V_{T}$ is piecewise-affine, the optimal inputs must occur at the boundary conditions, where $u =0$ or $u = x/2$ or $u = y$ or $u = x$. According to this, we can solve for the optimal policy algebraically to be
\begin{equation*}
    \pi^*_{T -1}(x, y) = \begin{cases}
        0 & \text{ if } x \leq 0, \\
        x & \text{ if } 0 \leq x \leq y, \\
        y & \text{ if }  y \leq x \leq 2 y \text{ or} \\
          &\text{ if }  2 y \leq x \text{ and } \\
          &(P + 2 - 3b)y \leq (\frac{P}{2} - \frac{3b}{2} + 1)x, \\
        \frac{x}{2} & \text{ otherwise}. \\
    \end{cases}
\end{equation*}

As we can see, even with a horizon of two steps, the optimal policy differs greatly from the original threshold policy in Eq. \eqref{eq:threshold}.
\end{ex}

While generating closed form expressions for optimal policies are hard to do in general (as seen in Example \ref{ex:peak}), if we limit to a horizon of $1$, we have that a modified threshold policy depicted in Figure \ref{fig:modpolicy} is optimal. We characterize the structure in the following Proposition.

\begin{figure}[ht]
    \centering
    \includegraphics[width=200pt]{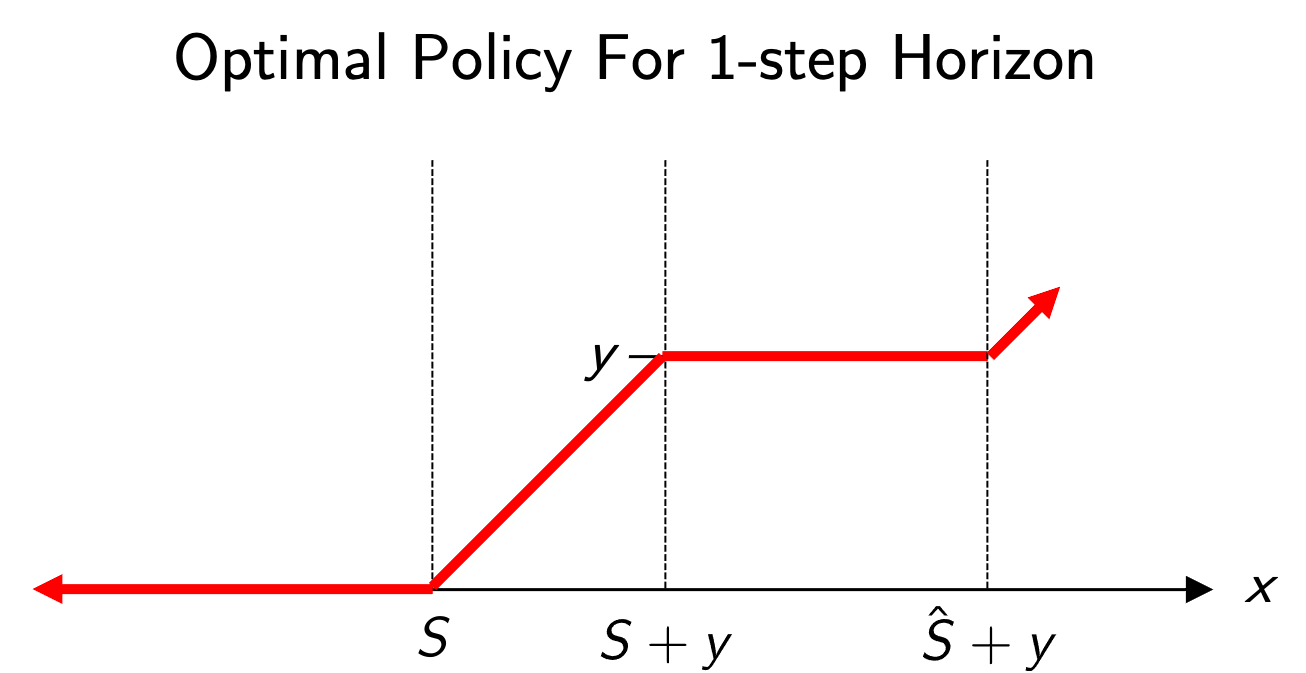}
    \caption{We depict the modified policy.}
    \label{fig:modpolicy}
\end{figure}

\begin{prop}
Consider the industrial refrigeration problem with a total cost in Eq. \eqref{eq:totalcost} with a horizon $T=1$ with no set up cost ($K = 0$ in Eq. \eqref{eq:ordering}). The optimal policy $\pi^*_1$ is characterized by
\begin{equation}
\label{eq:onestep}
    \pi^*_1(x, y) = 
    \begin{cases}
        x - \hat{S} & \text{ if } x \geq \hat{S} + y, \\
        y & \text{ if }  x \in [ S+y, \hat{S} + y], \\
        x - S & \text{ if } x \in [S, S+y], \\
        0 & \text{ if } x \leq S.
    \end{cases}
\end{equation}
for some $s \geq S \in \R_{\geq 0}$ and $ \hat{S} \geq S \in \R_{\geq 0}$.
\end{prop}
\begin{proof}
First, we express the value function $V_1$ as below
\begin{equation*}
    V_1(x, y) =
        \min_{u > 0} \Big\{f(x^+) + P y^+ \Big\} + ax,
\end{equation*}
with $f(z) \triangleq \bb{E}[h(z + q_t)] - a z$ being a convex scalar function. Let $S$ be the minimum of $f$. Note that if $x < S$, then by convexity of $f(z)$, the optimal input is $u = 0$. Additionally, if $S \leq x \leq S + y$, the optimal input is $u = x - S$. Now let $\hat{S}$ be such that the subderivative $\partial f/\partial z|_{z = \hat{S}} \ni P$. Since $f$ is convex, the minimum is defined by the first order condition on the subderivative $-\partial f/ \partial z|_{x - u} + P \ni 0$. This gives the first two conditions of the optimal policy.
\end{proof}

While we can compute the optimal one-step policy in closed form, this is not true in general, as shown in Example \ref{ex:peak}. However, as shown in the next Proposition, we can still derive a qualitative threshold-like characterization of the optimal policy in general. We depict the optimal policy pictorially in Figure \ref{fig:optgen}.

\begin{prop}
\label{prop:main}
    Consider the industrial refrigeration problem with a total cost in Eq. \eqref{eq:totalcost} with no set up cost ($K = 0$ in Eq. \eqref{eq:ordering}). The optimal policy $\pi^*_t$ is characterized by
    \begin{equation}
    \label{eq:optpolicy}
        \pi^*_t(x, y) \in 
        \begin{cases}
            [y, z^*_t] & \text{ if } g_t(y) + y \leq x, \\
            \{x - g_t(y) \} & \text{ if } x - y  \leq g_t(y) \leq x, \\
            \{0\} & \text{ if } x \leq g_t(y),
        \end{cases}
    \end{equation}
    where $g_t(y): \R_{\geq 0} \to \R$ and $z^*$ is such that $x = g_t(z^*_t) + z^*_t$.
\end{prop}
\begin{proof}
By assumption, we have that $K = 0$ in Eq. \eqref{eq:totalcost}. First we show that for every $t$, $V_t$ is non-decreasing in $y$ via backwards induction. For the base case, $V_{T+1}(x, y) = P y$ is non-decreasing in $y$, as $P > 0$. From Eq. \eqref{eq:VT}, we have that $V_{t+1}(x^+, y^+)$ is a composition of non-decreasing functions in $y$ by the induction assumption. Since monotonicity is preserved under expectations and infimum projections, the iterate $V_t$ is thus non-decreasing in $y$ as well, and we have the claim.

Now we show that $V_t$ is convex for every $t$. Similarly as before, we show this via backwards induction. For the base case, $V_{T+1}(x, y) = P y$ is affine in $y$ and thus convex. For the inductive case, we have that $\c$ is convex in $x$ and $u$. Additionally, $V_{t+1}$ is convex by the induction assumption. Moreover, since $x^+$ is affine in $x$ and $u$; $y^+$ is convex in $y$ and $u$; and $V_{t+1}$ is non-decreasing in $y$, $V_{t+1}(x^+, y^+)$ is a convex function of $x$, $y$, and $u$. Since convexity is preserved under expectation, $V_{t}(x, y)$ can be concisely expressed as
\begin{equation}
\label{eq:proofVt}
    V_t(x, y) = \min_{u \geq 0} f_t(x^+, y^+) + ax
\end{equation}
where $f_t(x, y) = \bb{E}_{q_t}[\h(x + q_t) + V_{t+1}(x + q_t, y)] - ax$ is a convex function. Since $\{u \in \R: u \geq 0\}$ is convex, convexity is preserved under the minimization, and thus $V_t(x, y)$ is a convex function. Thus the claim is shown.

Now we describe the optimal policy $\pi_t(x, y)$. Note that we can describe the optimal policy as
\begin{equation*}
    \pi_t(x, y) = \arg \min_{u \geq 0} \big\{ f_t(x^+, y^+) \big\} + a x.
\end{equation*}
where $f_t(x, y)$ is defined as in Eq. \eqref{eq:proofVt}. For each $y$, define the function $g_t(y)$ to be 
\begin{equation}
    g_t(y) = \max \Big\{ \arg \min_x f_t(x, y) \Big\}.
\end{equation}

We describe the three cases for the optimal policy as shown in Figure \ref{fig:optgen} and Eq. \eqref{eq:optpolicy}. If $x \leq g_t(y)$, then observe that $f_t(x^+, y^+) \geq f_t(x^+, y)$, since $V_{t+1}$ is non-decreasing in $y$. Additionally, as $f_t$ is convex, and $x \leq g_t(y)$, we have that $f_t(x^+, y) \geq f_t(x, y)$ for all $u \geq 0$, and thus $\pi^*_t(x, y) = 0$ when $x \leq g_t(y)$. Likewise, if $x - y  \leq g_t(y) \leq x$ we have that $f_t(x^+, y^+) \geq f_t(x^+, y) \geq f_t(g_t(y), y)$ for all $u \geq 0$, and thus $\pi_t(x, y) = x - g_t(y)$.

For the third condition $g_t(y) + y \leq x$, we first note that for $u = y$ generates a cost of $f_t(x-y, y) \leq f_t(x-u, y)$ for any $u \leq y$. Thus the optimal policy must satisfy $\pi^*(x, y) \geq y$. Moreover, let $u = z^*$ be the input in which $(x^+, y^+)$ intersects the curve $x = g_t(y)$. Note that $f(x^+, y^+) \geq f(x^+, z^*) \geq f(g_t(z^*), z^*)$ for any $u \geq z^*$, since $f$ is non-decreasing in $y$ and $f$ is convex. Thus the optimal policy must satisfy $\pi^*(x, y) \leq z^*$.
\end{proof}

\begin{figure}[ht]
    \centering
    \includegraphics[width=200pt]{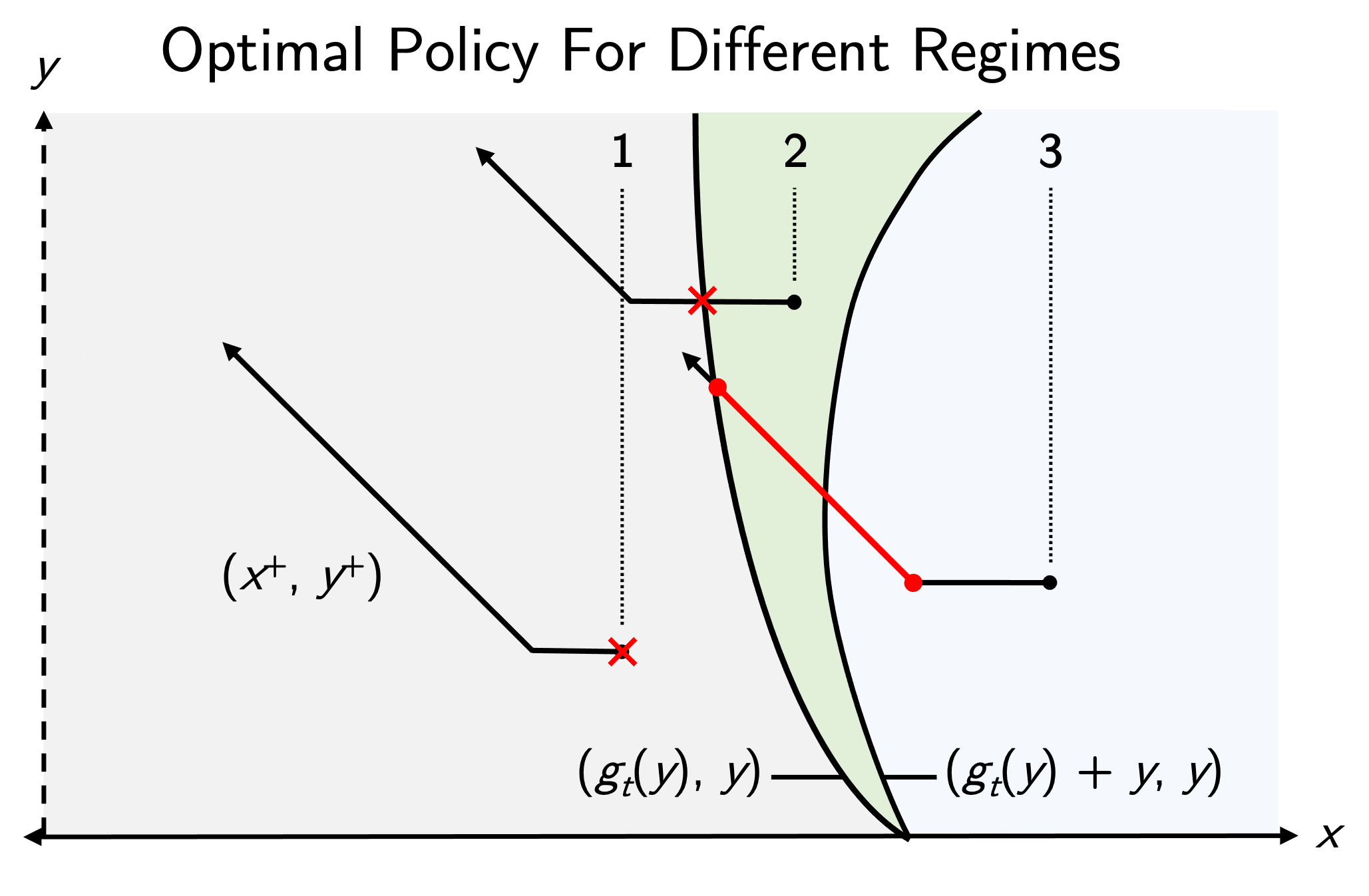}
    \caption{We depict the possible optimal inputs for each of the three regimes corresponding to Eq. \eqref{eq:optpolicy}, given the curve $g_t(y)$.}
    \label{fig:optgen}
\end{figure} 

\begin{rem*}
From Proposition \ref{prop:main}, we note that if $g_t(y) + y \leq x$, the optimal control action $u$ may lie anywhere in the interval $[y, z^*_t]$. Where on this interval depends non-trivially on the relative peak cost $P$ and the specific holding cost $\h$ and the current time $t$. However, we maintain a threshold-like policy structure, where below a buffer temperature $g_t(y)$ (dependent on the current peak), the optimal decision is to turn off the refrigeration system. 
\end{rem*}

\begin{ex}
The optimal policy $\pi^*(x, y)$ in Eq. \eqref{eq:optpolicy} may incur non-intuitive dependencies on the current peak value $y$; we outline a simplified scenario that outlines this phenomena. Consider a horizon of $T=2$, where the starting temperature is $x_1 = -2$ and the incoming heat is $q_t = 2$ for all $t$. Let there be no per-unit cost ($a = 0$) and setup cost ($K = 0$). Furthermore, we set the peak costs to $P \gg 1$ and the temperature penalty function to be 
\begin{equation*}
    \h(x) = \begin{cases}
        b x & \text{ if } x \geq 0 \\
        - x & \text{ if } x < 0,
    \end{cases}
\end{equation*}
where $b \gg 1$. We compare two scenarios: where the current peak is $y = 1$ and the current peak is $y = 2$. In the first case ($y = 1$), observe that the optimal control sequence is $u_1^* = 1$ and $u_2^* = 1$ to not incur any temperature violation costs. However, when $y = 2$, the optimal control sequence is $u_1^* = 0$ and $u_2^* = 2$. As such, the optimal policy $\pi_1(x, y)$ actually increases with the current peak $y$, counter to intuition that increasing that the peak should introduce more conservatism to the optimal policy. 
\end{ex}

\section{Simulation Study on Policy Designs}
\label{sec:simulations}

In this section, we evaluate the performance of different policy designs through a case scenario based on the thermal data (also represented in Figure \ref{fig:Qt}) acquired from a refrigeration facility owned by Butterball LLC\textregistered. We depict the possible incoming heat distributions during the on and off peak hours in Figure \ref{fig:onoff}.

\begin{figure}[ht]
    \centering
    \includegraphics[width=250pt]{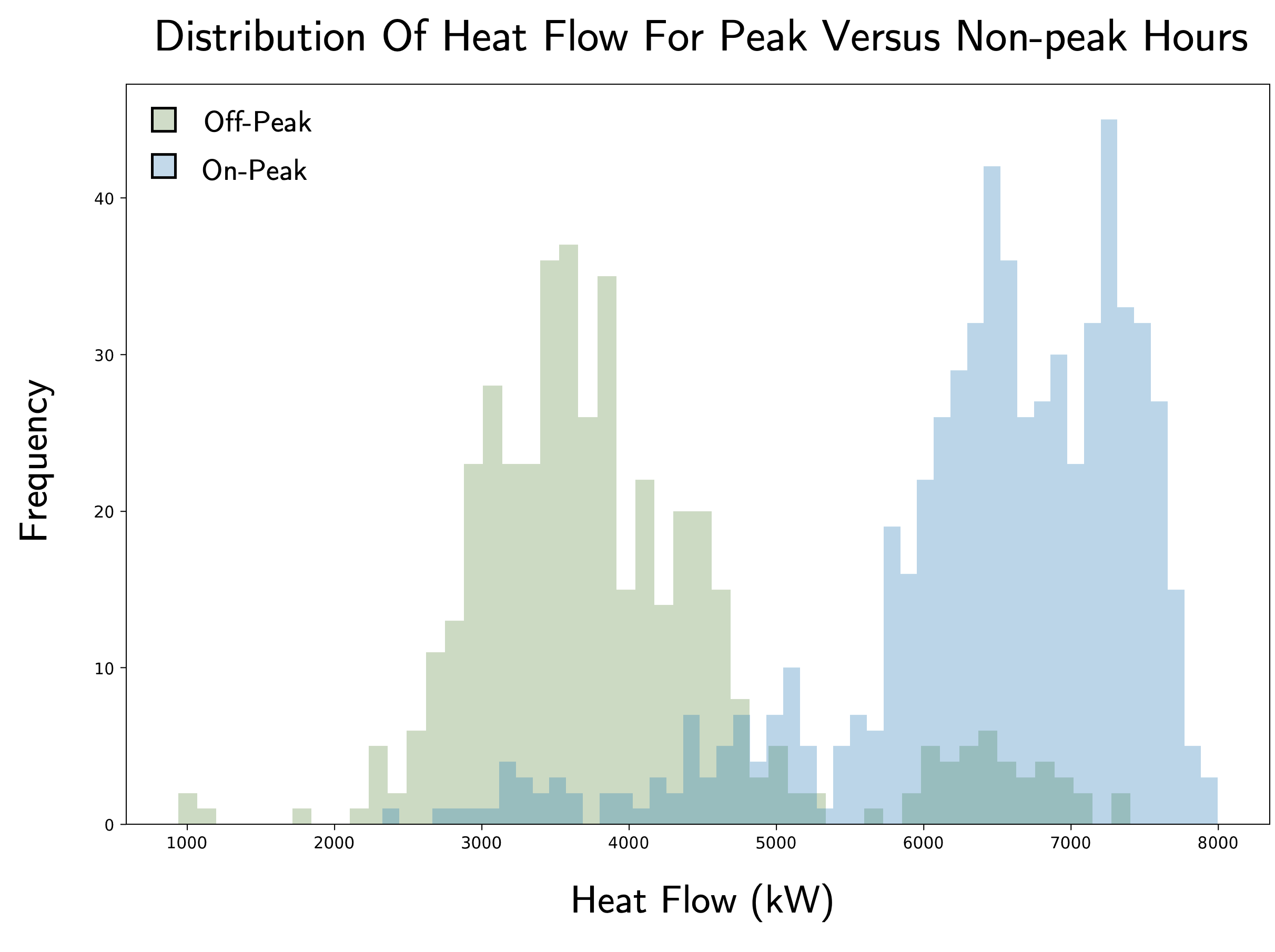}
    \caption{Incoming heat distributions of on and off peak time periods.}
    \label{fig:onoff}
\end{figure} 
 
The specifics used in the case study are outlined as follows. We evaluate the total cost as according to Eq. \eqref{eq:totalcost}, where each time step $t$ approximately corresponds to a $12$-hour time window and the trajectory cost is evaluated over horizon of a month. If the time $t$ corresponds to the on-peak, then the incoming heat distribution $\bf{Q}_t$ is exactly the one depicted in Figure \ref{fig:onoff} (and respectively for the off-peak time periods). We additionally normalize the incoming heat values to Megawatts to prevent numerical issues. We assume that the starting temperature and starting peak value are both $0$. We use the temperature violation cost of the form in Eq. \eqref{eq:viol} with $b_t = 20$ and $d_t = 1$. We additionally set the respective coefficients for peak costs to $P = 30$, setup costs to $K = 0$, and per-unit costs to $a = 1$ to reasonably model a possible refrigeration scenario.

\begin{table}[ht!]
\centering
\begin{tabular}{ | c | c |  }
\hline 
 Methodology & Total Cost  \\
 \hline 
 \hline
 Static Threshold Policy &  36.1 \\  
 Dynamic Threshold Policy & 26.7 \\ 
 Modified Threshold Policy & 30.7 \\ 
 Dynamic Modified Threshold Policy & 25.5 \\
 \hline 
\end{tabular}
\caption{Cost of Algorithms}
\label{table:main}
\end{table}

Under this example setup, we evaluate the performance of different types of threshold policies according to the total cost in Eq. \eqref{eq:totalcost}. We first examine the performance of the threshold policy in Eq. \eqref{eq:threshold}, where $s_t = S_t$ is not time varying. Additionally, we examine the dynamic version, where $s_t$ is allowed to depend on if $t$ corresponds to an on-peak or off-peak period. We also examine the performance of the policies stated in Proposition \ref{prop:nopeak} for time-dependent and time-independent values as well. For these policies, we optimize for the coefficients through a Monte Carlo search. The performance of these algorithms are depicted in Table \ref{table:main}.

\section{Conclusion}
\label{sec:conc}
Improving the control algorithms for industrial refrigeration systems can have significant impacts on energy usage, given their widespread presence and substantial scale. This work has thus focused on control algorithm design, specifically investigating the implications of various energy cost-rate structures, such as time-of-use and peak pricing. To do this, we leverage tools from inventory control to characterize the structure of the optimal policies. While simple threshold policies are optimal under solely time-of-use cost structures, adding significant peak costs can markedly change the policy structure of the optimal control algorithms. We provide characterizations of the one-step optimal policies in Proposition \ref{prop:nopeak} and generalize to arbitrary horizons in Proposition \ref{prop:main}, highlighting  significant differences from the classical threshold policies. Through theoretical analysis and simulations, we underscore the necessity for mindfulness in constructing control policies in the face of significant peak costs - thus providing valuable insights for more effective refrigeration management.

\bibliographystyle{ieeetr}
\bibliography{references.bib}

\begin{thebibliography}{10}

\bibitem{dincer2017refrigeration}
I.~Dincer, {\em Refrigeration systems and applications}.
\newblock John Wiley \& Sons, 2017.

\bibitem{stoecker1998industrial}
W.~F. Stoecker, {\em Industrial refrigeration handbook}.
\newblock McGraw-Hill Education, 1998.

\bibitem{fabrega2010exergetic}
F.~Fabrega, J.~Rossi, and J.~d'Angelo, ``Exergetic analysis of the refrigeration system in ethylene and propylene production process,'' {\em Energy}, vol.~35, no.~3, pp.~1224--1231, 2010.

\bibitem{tassou2010review}
S.~Tassou, J.~S. Lewis, Y.~Ge, A.~Hadawey, and I.~Chaer, ``A review of emerging technologies for food refrigeration applications,'' {\em Applied Thermal Engineering}, vol.~30, no.~4, pp.~263--276, 2010.

\bibitem{eiareport}
``Manufacturing energy consumption survey,'' tech. rep., US Energy Information Administration, 2018.

\bibitem{zhao2013model}
L.~Zhao, W.~Cai, X.~Ding, and W.~Chang, ``Model-based optimization for vapor compression refrigeration cycle,'' {\em Energy}, vol.~55, pp.~392--402, 2013.

\bibitem{larsen2003control}
L.~S. Larsen, C.~Thybo, J.~Stoustrup, and H.~Rasmussen, ``Control methods utilizing energy optimizing schemes in refrigeration systems,'' in {\em 2003 European Control Conference (ECC)}, pp.~1973--1977, IEEE, 2003.

\bibitem{larsen2006model}
L.~F.~S. Larsen, {\em Model based control of refrigeration systems}.
\newblock Department of Control Engineering, Aalborg University, 2006.

\bibitem{manske2000performance}
K.~A. Manske, ``Performance optimization of industrial refrigeration systems,'' 2000.

\bibitem{hovgaard2013nonconvex}
T.~G. Hovgaard, S.~Boyd, L.~F. Larsen, and J.~B. Jorgensen, ``Nonconvex model predictive control for commercial refrigeration,'' {\em International Journal of Control}, vol.~86, no.~8, pp.~1349--1366, 2013.

\bibitem{yin2018model}
X.-H. Yin and S.-Y. Li, ``Model predictive control for vapor compression cycle of refrigeration process,'' {\em International Journal of Automation and Computing}, vol.~15, no.~6, pp.~707--715, 2018.

\bibitem{shafiei2014model}
S.~E. Shafiei, J.~Stoustrup, and H.~Rasmussen, ``Model predictive control for flexible power consumption of large-scale refrigeration systems,'' in {\em 2014 American Control Conference}, pp.~412--417, IEEE, 2014.

\bibitem{sun2013peak}
Y.~Sun, S.~Wang, F.~Xiao, and D.~Gao, ``Peak load shifting control using different cold thermal energy storage facilities in commercial buildings: A review,'' {\em Energy conversion and management}, vol.~71, pp.~101--114, 2013.

\bibitem{yao2021state}
Y.~Yao and D.~K. Shekhar, ``State of the art review on model predictive control (mpc) in heating ventilation and air-conditioning (hvac) field,'' {\em Building and Environment}, vol.~200, p.~107952, 2021.

\bibitem{pattison2016optimal}
R.~Pattison, C.~Touretzky, T.~Johansson, I.~Harjunkoski, and M.~Baldea, ``Optimal process operations in fast-changing electricity markets: framework for scheduling with low-order dynamic models and an air separation application,'' {\em Industrial \& Engineering Chemistry Research}, vol.~55, pp.~4562--4584, 2016.

\bibitem{pattison2017moving}
R.~Pattison, C.~Touretzky, I.~Harjunkoski, and M.~Baldea, ``Moving horizon closed‐loop production scheduling using dynamic process models,'' {\em AIChE Journal}, vol.~63, pp.~639--651, 2017.

\bibitem{vishwanath2019iot}
A.~Vishwanath, V.~Chandan, and K.~Saurav, ``An iot-based data driven precooling solution for electricity cost savings in commercial buildings,'' {\em IEEE Internet of Things Journal}, vol.~6, no.~5, pp.~7337--7347, 2019.

\bibitem{risbeck2019economic}
M.~J. Risbeck and J.~B. Rawlings, ``Economic model predictive control for time-varying cost and peak demand charge optimization,'' {\em IEEE Transactions on Automatic Control}, vol.~65, no.~7, pp.~2957--2968, 2019.

\bibitem{mo2021optimal}
Y.~Mo, Q.~Lin, M.~Chen, and S.-Z.~J. Qin, ``Optimal online algorithms for peak-demand reduction maximization with energy storage,'' in {\em Proceedings of the twelfth ACM international conference on future energy systems}, pp.~73--83, 2021.

\bibitem{oldewurtel2010reducing}
F.~Oldewurtel, A.~Ulbig, A.~Parisio, G.~Andersson, and M.~Morari, ``Reducing peak electricity demand in building climate control using real-time pricing and model predictive control,'' in {\em 49th IEEE conference on decision and control (CDC)}, pp.~1927--1932, IEEE, 2010.

\bibitem{axsater2015inventory}
S.~Axs{\"a}ter, {\em Inventory control}, vol.~225.
\newblock Springer, 2015.

\bibitem{liu2012decision}
B.~Liu and A.~O. Esogbue, {\em Decision criteria and optimal inventory processes}, vol.~20.
\newblock Springer Science \& Business Media, 2012.

\bibitem{factandfig}
``Industrial refrigeration facts and figures.'' \url{https://bptraining.ornl.gov/wp-content/uploads/2021/05/Ind.-Refrigeration-Facts-Figures.pdf}.
\newblock Accessed: 2024-03.

\bibitem{porteus1990stochastic}
E.~L. Porteus, ``Stochastic inventory theory,'' {\em Handbooks in operations research and management science}, vol.~2, pp.~605--652, 1990.

\bibitem{lu2014inventory}
Y.~Lu and M.~Song, ``Inventory control with a fixed cost and a piecewise linear convex cost,'' {\em Production and Operations Management}, vol.~23, no.~11, pp.~1966--1984, 2014.

\bibitem{sobel1970making}
M.~J. Sobel, ``Making short-run changes in production when the employment level is fixed,'' {\em Operations Research}, vol.~18, no.~1, pp.~35--51, 1970.

\bibitem{scarf1960optimality}
H.~Scarf, K.~Arrow, S.~Karlin, and P.~Suppes, ``The optimality of (s, s) policies in the dynamic inventory problem,'' {\em Optimal pricing, inflation, and the cost of price adjustment}, pp.~49--56, 1960.

\bibitem{song2023research}
J.-S.~J. Song, {\em Research handbook on inventory management}.
\newblock Edward Elgar Publishing, 2023.

\end{thebibliography}
\end{document}